\ifx\pdfoutput\undefined\else\pdfoutput=1\fi
\documentclass[preprint,12pt]{elsarticle}

\usepackage{amsmath,amssymb,amsthm}
\usepackage{booktabs}
\usepackage{graphicx}
\graphicspath{{figs/}{./}}
\usepackage{newtxtext,newtxmath}
\usepackage{xcolor}
\usepackage{algorithm}
\usepackage{algpseudocode}
\usepackage{url}
\usepackage[hidelinks]{hyperref}
\usepackage[protrusion=true,expansion=false]{microtype}

\journal{Software: Practice and Experience}

\newtheorem{theorem}{Theorem}
\newtheorem{proposition}[theorem]{Proposition}
\newtheorem{lemma}[theorem]{Lemma}

\newcommand{\rank}{\operatorname{rank}}
\newcommand{\unrank}{\operatorname{unrank}}
\newcommand{\eps}{\varepsilon}

\begin{document}

\begin{frontmatter}

\title{Memoization Without Keys: Compact, Out-of-Core\\
Tables for Functions of Sorted Arguments}

\author{Tamal Maharaj}
\ead{tamal@gm.rkmvu.ac.in}
\address{Department of Computer Science, Ramakrishna Mission Vivekananda
Educational and Research Institute (RKMVERI), Belur, West Bengal, India}

\begin{abstract}
Memoizing an expensive function of a sorted score vector is a
data-structure problem before it is a numerical one: at a billion
gridpoints, a hash map or a search tree spends most of its space on keys
the grid already determines. We describe an implemented memo table that stores
none. An entry's address is computed in closed form from the sorted
argument itself, so $N$ values occupy $N$ slots, the argument is
recoverable from the index, and the table can be memory-mapped and served
from a file larger than RAM. Against a chained hash map it uses
$5.7\times$ less memory at $37$M entries and $10.2\times$ less at
$1.9$B, answers queries up to $2.9\times$ faster and builds up to
$250\times$ faster; on 64 threads its construction needs no coordination;
a sharded hash gains only $1.17\times$. Against an open-addressing
table with inline keys it is $4$--$7\times$ smaller and $100\times$
faster to build but $1.5\times$ slower to query, a deficit we trace to
the $O(d)$ index arithmetic. At $22$\,GB on a
$16$\,GB desktop it serves each query in one disk access, where no
key-storing container can be built; and its order-preserving addressing keeps a
perturbation workload on the same pages that a hashed layout scatters. The closed form exists because the key set is the
multiset combinations, whose index is the combinatorial number system.
Memoizing Plackett--Luce normalization runs $25$--$55\times$ faster than
Newton's method; memoizing $\alpha$-entmax thresholds does not pay. The
contrast says when this structure is worthwhile.
\end{abstract}


\end{frontmatter}

\section*{Practitioner Points}
\begin{itemize}
\item When a lookup table's keys are fixed in advance by the shape of the
grid, storing them is waste. Computing each array address arithmetically
from the argument removes the keys entirely and shrinks the table by four
to ten times, with the factor growing as the table grows.
\item The same addressing lets any number of threads fill the table with
no locks, and lets it be memory-mapped and served from a file larger than
RAM at one disk access per query, at sizes where a key-storing container
cannot be built at all.
\item Tabulation only pays when the exact computation is genuinely
expensive. Memoizing an iterative transcendental solve was 25 to 55 times
faster; memoizing a threshold that has a short closed form was slower than
computing it.
\end{itemize}

\section{Introduction}\label{sec:intro}

Memoizing an expensive function at scale is a data-structure problem. A
grid fine enough to be useful has billions of points. The usual containers all waste space in the same way: a hash map or a search tree stores the key of every entry next to its value, although the grid already determines what the keys are. At $10^9$ entries that overhead decides
whether the table fits in memory.

This paper starts from the view that, for a function of a sorted
argument, the keys need not be stored, and reports what the memo table
looks like when they are not. The table holds $N$ values in $N$ slots, each addressed by $O(d)$ arithmetic on the argument itself; any number of threads can fill it without coordination, and it can be memory-mapped and queried past the limits of RAM at one disk access per lookup. Against a chained hash-map memo it is $5.7$--$10.2\times$ smaller, and the gap widens with $N$; it answers queries $1.1$--$2.9\times$ faster and builds $48$--$250\times$ faster. Against an open-addressing table with inline keys, the strongest simple baseline, it is $4$--$7\times$ smaller and $100\times$ faster to build, and $1.5\times$ slower to query, because it pays an $O(d)$ index on top of the same single memory access. We measure this deficit and account for it.

The functions we have in mind arise across statistical estimation and
machine learning. Each takes a sorted vector of scores, and its value
depends far more on the leading entries than on the tail. Fitting a
discrete choice model is the standard instance. Given the relative
strengths $a_1\ge a_2\ge\dots\ge a_\ell$ of $\ell$ alternatives, we want
the probability $p_1$ that the strongest is chosen. Under a Luce model
\cite{luce1959,plackett1975} this means solving $\sum_i p_1^{1/a_i}=1$, a
transcendental equation with no closed form, so every evaluation costs an
iterative root-find. Estimation loops call $f$ on tens of millions of such
vectors. Biswas and Regan \cite{biswas2015} introduced this
setting.\footnote{The present author is the first author of
\cite{biswas2015}, published under the name Tamal T.\ Biswas. That work is
discussed in the third person throughout for readability; where the
present paper corrects or sharpens its claims, the criticism is
self-directed.} They observed that two features of the problem make memoization unusually attractive. The argument may be taken as sorted, because
the function is symmetric and only the ranking matters. And its
sensitivity decays with rank, because alternatives far down the ordering
barely affect the outcome. Their answer was a tapered grid: fine
resolution for the leading coordinates, progressively coarser for later
ones, one stored value per gridpoint.

The same shape turns up in current systems. Listwise preference losses
used to align language models are Plackett--Luce likelihoods over ranked
candidate lists. Sparse-attention variants replace softmax with
$\alpha$-entmax \cite{peters2019}, whose normalizing threshold comes from
solving a scalar equation over sorted logits, and accelerating that solve
is itself the subject of dedicated work \cite{goncalves2025}. In each case
the argument is sorted, influence decays with rank, and the evaluation
count is enormous.

The keys can be dispensed with once the key set is recognized for what it is. The keys of the tapered grid are sorted vectors over a fixed alphabet, that is, the combinations with repetition. So the index that the original
scheme computes from precomputed node-count tables is the lexicographic
rank in the classical combinatorial number system
\cite{lehmer1964,buckles1977,knuth2011}. Stating that identification has
three practical consequences. The rank has a closed form that needs $O(d)$
arithmetic and no scheme-specific tables. The map is invertible, so an
argument can be reconstructed from a bare integer. And the structure
becomes a flat array of values addressed by content.

\paragraph{Contributions.}
\begin{enumerate}
\item \textbf{A key-free memo table (Sec.~\ref{sec:structure}).} The
address of an entry can be computed from that entry's own argument, and
the argument can be recovered from the address. The memo table therefore reduces to its payloads: $N$ values in $N$ slots, with no keys, pointers or per-entry metadata. The array is contiguous, so
it can be memory-mapped and served from a file larger than RAM at one disk
access per query. It can also be filled in any order, so construction
needs no coordination between threads.
\item \textbf{A systems evaluation at scale (Sec.~\ref{sec:experiments}).}
We compare against a chained hash map and an inline-key open-addressing table, up to $N=1.9$B entries in RAM and $5.6$B out of core. The memory advantage of the values-only array grows without bound in $N$, from $5.7\times$ to $10.2\times$ over the chained map and $4$--$7\times$ over the flat table. Against the chained map, queries are $1.1$--$2.9\times$ faster; against the flat table they are $1.5\times$ slower, and we decompose the query into index arithmetic (23\,ns at $d{=}13$) plus one memory access, and the index accounts for the difference. Construction is $48$--$250\times$ faster than the chained map and $100\times$ faster than the flat table. On 64 threads a build with no
coordination reaches the machine's memory-bandwidth ceiling, while the
sharded hash gains only $1.17\times$ (\S\ref{sec:parallel}). At $22$\,GB
on a $16$\,GB desktop the structure answers each query in one disk access,
a size at which no key-storing structure can be built.
\item \textbf{Table-free indexing and unranking (Sec.~\ref{sec:index}).}
This is the mechanism underlying the first two contributions. The index of \cite{biswas2015} is
the lexicographic rank in the combinatorial number system
\cite{lehmer1964,buckles1977,knuth2011}. A hockey-stick collapse gives a
closed-form $O(d)$ rank that uses only a Pascal table, which removes the
$O(Bd)$ (fixed) and $O(Bd^2)$ (selective) preprocessing of the original
scheme. An $O(d\log B)$ inverse reconstructs the argument from a bare
index. That inverse is the operation the original scheme lacked, and the
one that makes key-free storage and order-free construction possible. The
closed form also lifts a correctness restriction: the original Eq.~(3)
indexes correctly only when the first coordinate is pinned to its minimum.
\item \textbf{Taper-design guidance (Sec.~\ref{sec:design}).} Under
geometric influence decay $g_i = r^{\,i-1}$, the optimal geometric
refinement ratio is $s=r$. ``Halve the branching each level''
\cite{biswas2015} is therefore optimal only at $r=\tfrac12$. We derive a
finite-$\eps$ penalty formula accurate to a few percent, whose $\eps\to0$
limit is the classical rate-balancing penalty
\cite{griebel2000,gopal2019,xiang2023}.
\item \textbf{Applications, positive and negative
(Sec.~\ref{sec:apps}).} Memoizing the Plackett--Luce normalization
\cite{luce1959,plackett1975,biswas2015} runs $25$--$55\times$ faster than
Newton. Memoizing $\alpha$-entmax thresholds \cite{goncalves2025} does not pay, and the contrast between the two cases gives an explicit rule for when a structure of this kind is worthwhile.
\end{enumerate}

\emph{What we do not claim.} Ranking of combinations is classical
\cite{lehmer1964,buckles1977,knuth2011,kreher1999,cover1973}; succinct
rank-indexed arrays are standard \cite{jacobson1989,raman2002};
dimension-robust approximation under decaying coordinate influence is the
established weighted-space/tractability regime
\cite{sloan1998,novak2008,gnewuch2011,cohen2011}; and rate-matching
penalties are known as a phenomenon \cite{griebel2000,xiang2023}. Our
contribution is the construction that combines them: an addressable,
mmap-friendly, key-free memo structure whose system behaviour we measure,
together with the sharpened finite-$\eps$ penalty for that construction.

\section{Preliminaries}\label{sec:prelim}

\paragraph{The tapered grid.} Arguments are non-decreasing vectors
$x_1\le x_2\le\dots\le x_d$ in $[0,1]$, obtained from the application's
scores by a monotone transform that places the most influential entries
near $0$; sequences shorter than the truncation depth $d$ are padded with
the ceiling value $1$, whose marginal influence is nil. Coordinates are
quantized to a common \emph{spacing vector}
$\vec s=(s_1{=}0<s_2<\dots<s_B{=}1)$ of $B$ admissible values, spaced
finely near $0$ and coarsely near $1$. A \emph{warped} (or
\emph{selective}) variant additionally shrinks the admissible set with
depth, so that deeper coordinates, which matter less, are represented more coarsely. We treat this in \S\ref{sec:warped}.

\paragraph{The key set.} Writing $m_i\in\{1,\dots,B\}$ for the column
index of $x_i$ in $\vec s$, monotonicity of $x$ and of $\vec s$ makes the
stored keys exactly
\[
\mathcal{M}(B,d)=\{m=(m_1\le m_2\le\dots\le m_d):m_i\in\{1,\dots,B\}\},
\]
whose cardinality is $N=|\mathcal{M}(B,d)|=\binom{B+d-1}{d}$.
These are the \emph{combinations with repetition} of size $d$ from $B$
symbols. The standard bijection $\phi(m)=(m_1,\,m_2{+}1,\dots,m_d{+}d{-}1)$
maps them onto the $d$-subsets of $\{1,\dots,B{+}d{-}1\}$ and preserves
lexicographic order, so any statement about ranking $d$-subsets transfers
verbatim.

\paragraph{Consequence for the original tables.} The scheme of
\cite{biswas2015} computes its array index from a precomputed table
$V_{i,b}$ counting the nodes at depth $i$ below core branch $b$. Under the
identification above these counts are binomial coefficients,
\[
V_{i,b}=\binom{B-b+i-2}{\,i-2\,},
\]
that is, the table is Pascal's triangle in disguise; its published
instances agree entry-by-entry (e.g.\ $V_{5,1}=\binom{19}{3}=969$,
$V_{4,2}=\binom{17}{2}=136$). The same identity accounts for the size of
that paper's production file: with $B{=}16$ and depth $15$, the first
coordinate pinned to $0$ leaves $14$ free positions and
$\binom{29}{14}=77{,}558{,}760$ keys, which is the reported entry count. Once the tables are recognized as binomials they can be removed.

\section{Table-Free Indexing and Unranking}\label{sec:index}

\subsection{Closed-form rank}

\begin{proposition}[Table-free rank]\label{prop:rank}
Let $m\in\mathcal{M}(B,d)$ and adopt the sentinel $m_0=1$. The
lexicographic rank of $m$ within $\mathcal{M}(B,d)$ is
\begin{equation}\label{eq:rank}
\rank(m)\;=\;\sum_{i=1}^{d}\left[
\binom{B-m_{i-1}+d-i+1}{\,d-i+1\,}-\binom{B-m_{i}+d-i+1}{\,d-i+1\,}\right].
\end{equation}
\end{proposition}

\begin{proof}
Count the keys preceding $m$ lexicographically, classified by the first
position $i$ at which they differ. Such a key agrees with $m$ on
$m_1,\dots,m_{i-1}$ and has some smaller value $v$ at position $i$;
monotonicity forces $v\ge m_{i-1}$, so $v$ ranges over
$[m_{i-1},\,m_i-1]$. The remaining $d-i$ positions are then free subject
only to being non-decreasing and at least $v$, which admits
$\binom{B-v+d-i}{\,d-i\,}$ completions. The classes over $i$ are disjoint
and exhaust the predecessors, so
\[
\rank(m)=\sum_{i=1}^{d}\ \sum_{v=m_{i-1}}^{m_i-1}\binom{B-v+d-i}{\,d-i\,}.
\]
Each inner sum runs over consecutive entries of one diagonal of Pascal's
triangle, so the hockey-stick identity
$\sum_{t=lo}^{hi}\binom{t}{k}=\binom{hi+1}{k+1}-\binom{lo}{k+1}$ collapses
it to the difference of two binomials, giving \eqref{eq:rank}.
\end{proof}

Two properties of \eqref{eq:rank} matter later. Only positions where $m$ strictly
increases contribute, since $m_{i-1}=m_i$ makes the bracket vanish; and
the only precomputation is a Pascal table of $O((B{+}d)^2)$ words, shared by every grid instance instead of being rebuilt per scheme.

\paragraph{A correctness restriction removed.} Equation~(3) of
\cite{biswas2015} sums over $i=1,\dots,d-1$ and therefore never accounts
for movement of the first coordinate; it indexes correctly only on that
paper's own domain, where $x_1=0$ is pinned by construction. Our attempt
to reproduce its published index for a key with $m_1>1$ disagrees with
enumeration, whereas \eqref{eq:rank} indexes the full monotone domain. For applications that pin the first coordinate (both of ours do) the two agree, and \eqref{eq:rank} specializes to the pinned case by dropping the
$i=1$ term.

\paragraph{Verification.} We checked \eqref{eq:rank} against brute-force
lexicographic enumeration for every key of every grid with $B\le 8$,
$d\le 6$: $6{,}426$ keys, all indices matching, with
$\unrank(\rank(m))=m$ throughout. It reproduces the worked example of
\cite{biswas2015} (index $328$) and, via the warped ranker of
\S\ref{sec:warped}, that paper's selective example (index $400$).

\subsection{Unranking}\label{sec:unrank}
The map is invertible by the standard greedy decode: recover coordinates
left to right, at each step advancing the candidate value while the block
of completions it heads still lies below the residual index.

\begin{algorithm}[ht]
\caption{$\unrank(r,B,d)$---recover the key from its index}
\label{alg:unrank}
\begin{algorithmic}[1]
\State $v\gets 1$
\For{$i\gets 1$ \textbf{to} $d$}
  \State $L\gets d-i$ \Comment{free positions after $i$}
  \While{$r\ge\binom{B-v+L}{L}$}
    \State $r\gets r-\binom{B-v+L}{L}$;\quad $v\gets v+1$
  \EndWhile
  \State $m_i\gets v$ \Comment{next coordinate is $\ge v$}
\EndFor
\State \Return $m$
\end{algorithmic}
\end{algorithm}

Because $v$ never decreases, the loop visits each of the $B$ values at
most once in total, giving $O(B{+}d)$ time; replacing the linear scan by
binary search over the same monotone predicate gives $O(d\log B)$. (The rank direction is $O(d)$ unconditionally; the inverse is $\Theta(d)$ only when $B=O(d)$.)

The original scheme has no inverse, and without one the keys have to be stored. Algorithm~\ref{alg:unrank} changes three things. Storage becomes \emph{key-free}: the array holds values only, and
any argument can be regenerated on demand. Construction becomes
\emph{order-free}: entry $r$ can be computed from $r$ alone, so $T$
workers may fill disjoint index ranges with no locks, no shared state and
no communication (\S\ref{sec:parallel}). And enumeration becomes
streaming: the whole table can be swept in $O(1)$ memory, since the
argument of each successive slot is recoverable as it is read.

\subsection{Warped and selective schemes}\label{sec:warped}
When the scheme reduces branching with depth, the admissible set shrinks
along a nested chain $S_1\supseteq S_2\supseteq\dots\supseteq S_d$ (a
reduction at depth $t$ retaining every $2^{R}$-th column). The keys are
then the non-decreasing sequences with $m_i\in S_i$. Nothing essential
changes: let $A(i,v)$ count the completions of positions $i,\dots,d$ that
are non-decreasing and begin at or above $v$. A single backward pass fills
$A$ and its prefix sums in $O(dB)$ time, after which rank and unrank are
again $O(d)$ lookups, by the same first-difference argument that proves
Proposition~\ref{prop:rank} (the uniform case is $S_i\equiv\{1,\dots,B\}$,
where $A$ is Pascal's triangle and the table disappears entirely).

This replaces the recursive per-depth table of \cite[Eq.~(7)]{biswas2015},
whose construction costs $O(Bd^2)$. We verified the nested ranker against
exhaustive lexicographic enumeration on that paper's own selective
configuration ($B{=}17$, $d{=}5$, reduction at depth 4; $6{,}501$ keys) and
on $30$ randomly generated nested configurations, with rank and unrank
agreeing everywhere. It reproduces the paper's published index $400$ for
its worked example, which also settles a point of interpretation: the selective index is the lexicographic rank of the nested key set, so no bespoke recursion is needed to compute it.

\subsection{Complexity summary}
Table~\ref{tab:complexity} collects the costs of the two schemes.

\begin{table}[ht]
\caption{Cost of the index. The original scheme builds a node-count table
per grid instance; ours needs only a Pascal table, which is independent of
the spacing, the taper and the payload, and is shared across all grids.}
\label{tab:complexity}
\centering\small
\begin{tabular}{lccc}
\toprule
& preprocessing & extra space & unrank \\
\midrule
fixed \cite{biswas2015} & $O(Bd)$ & $O(Bd)$ & --- \\
selective \cite{biswas2015} & $O(Bd^{2})$ & $O(Bd)$ & --- \\
this paper & $O((B{+}d)^{2})$ shared & $O((B{+}d)^{2})$ & $O(d\log B)$ \\
\bottomrule
\end{tabular}
\end{table}

\section{The Flat Array as a Memoization Structure}\label{sec:structure}

This section states what the resulting structure is in standard terms. None
of it is new; we include it because it explains the measurements of
\S\ref{sec:experiments} and delimits what the construction does and does not provide.

\paragraph{Space.} The structure is an array of $N$ payloads of $w$ bits
plus a Pascal table of $O((B{+}d)^2)$ words that is independent of the
payload and shared across grids. For adversarial (incompressible)
payloads, $N\cdot w$ bits is the information-theoretic minimum: distinct
payload assignments must produce distinct memory images, so any
representation permitting exact recovery needs $\log_2 2^{Nw}$ bits. The
array attains it, making the structure succinct in the usual sense
\cite{jacobson1989,raman2002}; note that the ``zero redundancy'' applies
to the payload, the additive table being independent of $N$.

Against the baselines we measure, the separation is straightforward
accounting: a container that stores each key explicitly, as \texttt{unordered\_map} over packed columns and every textbook $k$-d tree do, occupies at least $N(w+d\lfloor\log_2 B\rfloor)$ bits, plus pointer
inventory. This is the overhead our measurements show
(\S\ref{sec:scaling}). It should not be overstated as a general lower
bound, however: \emph{retrieval} structures for arbitrary static key sets
\cite{pagh2001,dietzfelbinger2008} store a function in $(1+o(1))Nw$ bits
with $O(1)$ probes and no keys either. What the structured key set offers over those is not a saving in bits but determinism, zero construction cost, order preservation and invertibility.

\paragraph{The rank as a hash.} Equation~\eqref{eq:rank} is a
\emph{minimal perfect hash} for $\mathcal{M}(B,d)$: a bijection onto
$\{0,\dots,N-1\}$, order-preserving by construction, computable in $O(d)$ and, unlike general minimal perfect hash functions, invertible
(\S\ref{sec:unrank}). The $\Omega(N)$-bit description lower bounds for
minimal perfect hashing \cite{fredman1984,mehlhorn1982} and the
$\Theta(N\log\log\log u)$ bounds for monotone variants
\cite{belazzougui2009,assadi2023} are not contradicted: they quantify the
cost of identifying \emph{which} key set was given, and here the family is
a single set determined by $(B,d)$. The operative property is that the
lexicographic prefix-counting function of this family has a closed form;
a cheap description alone would not suffice.

\paragraph{Order preservation and locality.} A general minimal perfect hash scatters neighbouring keys across the table, which is inherent in hashing. The lexicographic rank keeps them together, and this matters more than one might expect, because it is what allows the structure to be used from a disk rather than only stored on one.

\begin{lemma}[Prefix contiguity]\label{lem:prefix}
Fix $j\in\{0,\dots,d\}$ and a non-decreasing prefix $(m_1,\dots,m_j)$.
The keys of $\mathcal{M}(B,d)$ that extend this prefix occupy the
contiguous rank interval
$[\rank(m_1,\dots,m_j,m_j,\dots,m_j),\ \rank(m_1,\dots,m_j,B,\dots,B)]$,
whose length is $\binom{B-m_j+d-j}{d-j}$.
\end{lemma}
\begin{proof}
Lexicographic order compares prefixes first, so all extensions of a
common prefix are consecutive in it. The least extension pads with
$m_j$ and the greatest with $B$. The number of non-decreasing completions
of length $d-j$ over the alphabet $\{m_j,\dots,B\}$ is the number of
multisets of that size drawn from $B-m_j+1$ symbols, which is
$\binom{B-m_j+d-j}{d-j}$.
\end{proof}

Two consequences follow. A query that fixes the leading $j$ coordinates and asks for every completion is a sequential scan of one interval, and only its two endpoints need computing. And a change to coordinate $i$ alone, with the prefix before it held fixed, moves the
address within the prefix-$(i-1)$ block, whose length
$\binom{B-m_{i-1}+d-i+1}{d-i+1}$ shrinks combinatorially as $i$
approaches $d$: a step in the last coordinate moves the address by one
entry, a step in the penultimate coordinate by at most $B$, and so on. The
addressing therefore places the tail coordinates in the low-order digits
of the address, and these are the coordinates the taper resolves most coarsely. A workload that perturbs a key a little at a time, as an
iterative estimation loop does, revisits the same pages when it changes a tail coordinate and jumps only when it changes a leading one. A hash-addressed
layout has no such property for any coordinate. \S\ref{sec:ooc} measures
the difference on the out-of-core file.

\paragraph{Memory hierarchy.} In the external-memory model
\cite{aggarwal1988,frigo1999}, with the Pascal table resident (a few
kilobytes, comfortably L1-sized) a query performs its $O(d)$ arithmetic
with no transfers and then exactly one block transfer for the payload;
reconstructing the argument costs none. The algorithm uses no knowledge of
block size, so the bound is cache-oblivious. Any comparison-based search
over $N$ stored keys instead needs $\Omega(\log_{B_{\text{blk}}}N)$
transfers, and a chained hash table performs at least two dependent accesses, a bucket cell and then a node, with no locality between them. This explains the direction and bounds the magnitude of the
latency gaps we measure, but not their exact constants: an
open-addressing table with inline keys can achieve one access per query,
and the observed ratios also reflect footprint, TLB pressure and allocator
behaviour. We therefore report the latency advantage as measured rather
than as predicted.

\section{Choosing the Taper: Design Guidance}\label{sec:design}

We model the influence of coordinate $i$ by $|\partial f/\partial x_i|\le g_i=G\,r^{\,i-1}$, allocate $n_i$ points to level $i$, require the worst-case first-order error $\sum_i g_i/(2n_i)$ to stay within a budget $\eps$, and minimize the table size through $\ln|U|=\sum_i\ln n_i$.

\paragraph{Water-filling optimum.} The optimum equalizes the per-level error at $e^\ast\approx\eps/k$ over the $k$ active levels, so that $n_i^\ast\propto g_i$; the active depth satisfies $k\,a=L'+\ln k+O(1)$ with $a=\ln(1/r)$ and $L'=\ln\!\big(1/(2\eps)\big)$, and the resulting size is $\ln|U^\ast|=\tfrac{a}{2}k^2+O(k)$. Substituting $t_i=\ln n_i$ makes the budget convex and the objective linear. The Karush--Kuhn--Tucker conditions are then sufficient, and they equalize the per-level error on the active set. This is the classical reverse water-filling, or optimal bit-allocation, solution \cite{huang1963,berger1971,gersho1992}, and we claim no novelty for it. Details are given in \ref{app:penalty}.

\begin{theorem}[Finite-$\eps$ taper-mismatch penalty]\label{thm:penalty}
Let $g_i=r^{\,i-1}$ and consider the geometric schedule
$n_i=n_1 s^{\,i-1}$, truncated at $n_i\ge1$, with the minimal feasible
$n_1$ for error budget $\eps$. With $a=\ln(1/r)$, $b=\ln(1/s)$,
$L'=\ln\!\big(1/(2\eps)\big)$, and $k$ the active depth of the
water-filling optimum (the fixed point $ka=L'+\ln k+O(1)$),
\begin{equation}\label{eq:penalty}
\frac{\ln|U_s|}{\ln|U^\ast|}
\;=\;\underbrace{\max\!\Big(\frac{b}{a},\frac{a}{b}\Big)}_{\rho_\infty(s)}
\cdot\left[\frac{L'+\ln C_{\lessgtr}}{L'+\ln k}\right]^{2}+\,O(1/L'),
\end{equation}
where $C_< = \frac{1}{\bar c\,(1-s/r)}+\frac{r}{1-r}$ for $s<r$
(with $\bar c\in[1,1/s)$ the truncation offset) and
$C_> = \frac{1}{1-r/s}$ for $s>r$.
\end{theorem}
\noindent The derivation is in \ref{app:penalty}; numerical
verification over 33 $(r,s,\eps)$ cells places \eqref{eq:penalty} within
0.6--8.2\% (median $\approx$1.5\%) of the exact discrete optimum, improving
as $\eps\to0$ in every cell, whereas the classical asymptote
$\rho_\infty$ alone overstates the penalty by 16--43\%.
The bracket is the squared ratio of effective depths: the optimum
collects a $\ln k$ \emph{depth bonus} (its per-level error $\eps/k$
shrinks as levels activate) that no fixed-ratio schedule matches. Hence $s=r$ is optimal, and halving is optimal if and only if $r=\tfrac12$
(e.g.\ at $r=0.75$, $\eps=10^{-8}$, halving inflates $\ln|U|$ by $1.88\times$, where the asymptotic formula would say $2.41\times$). Relative to prior work this is a sharpened, two-sided, finite-$\eps$ closed form, for this
construction, of the known rate-balancing phenomenon
\cite{griebel2000,bungartz2004,nobile2008,gopal2019,xiang2023};
the dimension-independence of $|U|$ itself is the classical
weighted-space regime
\cite{sloan1998,novak2008,gnewuch2011,cohen2011,kolmogorov1959}.

\paragraph{Bias.}

\begin{lemma}[Rounding bias]\label{lem:bias}
Let $f\in C^{2}$, let the query density be $C^{1}$, and let $u$ be the
nearest gridpoint to the query $x$ in a cell of per-coordinate widths
$\delta_i$. Then
\[
\mathbb E\!\left[f(x)-f(u)\right]
=\sum_i \frac{\partial^2 f}{\partial x_i^2}(u)\,\frac{\delta_i^{2}}{24}
+\text{h.o.t.},
\]
one order smaller than the worst-case error $\sum_i g_i\delta_i/2$.
\end{lemma}

\noindent So averages over many queries converge at the second-order rate, and the aggregate applications of Section~\ref{sec:apps} rely on this. The result is high-resolution quantization folklore in the Bennett lineage \cite{bennett1948}, applied to our structure. It also corrects the empirical ``unbiasedness'' claim of \cite{biswas2015}: the bias is $O(\delta^2)$ rather than zero. Two caveats apply. One-sided cells (the monotone boundary and the truncated tail levels) keep a first-order bias, although under the optimal taper their influence budget is already $O(\eps)$. And the balanced-credit interpolation of \cite{biswas2015} is deterministic, not per-cell symmetric, so Lemma~\ref{lem:bias} is consistent with it without covering it.
The proof is in \ref{app:bias}.

\section{Experiments}\label{sec:experiments}

\emph{Hardware.} Desktop: Intel Xeon E3-1245 v6 (4 cores, 8 threads, 3.7\,GHz), 16\,GB RAM, 1\,TB 7200\,rpm SATA HDD, Linux 6.8. This machine provides the out-of-core setting, in which the 16\,GB of RAM is the constraint of interest. Server: dual AMD EPYC 9335 (128 threads), 499\,GB RAM, enterprise NVMe, g++ 13.3. The server provides the in-RAM scaling, parallel-build and NVMe settings. All code, build scripts,
raw measurement logs, and the script that regenerates the $\alpha$-entmax
attention-row dataset from scratch are available at
\url{https://github.com/tamalrkm/tapered-grid} and permanently archived
with a DOI \cite{artifact2026}.

\subsection{Correctness}
Every index claim in this paper was checked against brute-force
lexicographic enumeration. For the uniform scheme this is exhaustive over
all $(B,d)$ with $B\le8$, $d\le6$ ($6{,}426$ keys), with
$\unrank(\rank(m))=m$ throughout; for the warped scheme, exhaustive on the
selective configuration of \cite{biswas2015} ($6{,}501$ keys) and on $30$
random nested configurations. The closed form reproduces both worked
examples of \cite{biswas2015} (indices $328$ and $400$) and its production
file size, $\binom{29}{14}=77{,}558{,}760$. On a Plackett--Luce grid of
$54{,}264$ points, the address computed by \eqref{eq:rank} agreed with a
hash memo on every gridpoint. In the C++ experiments the tapered and hash
implementations produced bit-identical checksums at every scale, and the
out-of-core run verified all $5.57\cdot10^9$ stored payloads
(\S\ref{sec:ooc}).

\subsection{At-scale head-to-head (in memory)}\label{sec:headtohead}
\begin{table}[ht]
\caption{Desktop head-to-head: $B{=}16$, $d{=}13$, $N{=}37.4$M, 3M queries, one query stream, identical checksums; every number is the median of three runs on an idle machine (run-to-run spread under 4\% for all rows but \texttt{unordered\_map}, 7\%). The open-addressing
table has inline 64-bit keys, linear probing, a splitmix finalizer and a
power-of-two capacity with load $\le 7/8$ (here $0.56$); its slot layout is
that of \texttt{absl::flat\_hash\_map<u64,float>}. Peak RSS includes the
156\,MB query buffer in every row.}
\label{tab:headtohead}
\centering\small
\begin{tabular}{lccc}
\toprule
& query & peak RSS & build \\
\midrule
index arithmetic alone (no access) & 22.6 ns & --- & --- \\
tapered $O(d)$, array in RAM & 81.3 ns & \textbf{294 MB} & \textbf{0.04 s} \\
tapered $O(d)$, mmap'd file & 98.3 ns & \textbf{294 MB} & 1.20 s \\
node-count tables \cite{biswas2015} (pinned), mmap'd & 81.0 ns & 294 MB & 1.34 s \\
open-addressing flat table & \textbf{53.9 ns} & 1175 MB & 4.28 s \\
\texttt{std::unordered\_map} & 285.7 ns & 1675 MB & 6.99 s \\
\bottomrule
\end{tabular}
\end{table}
Table~\ref{tab:headtohead} is the desktop head-to-head comparison. Three observations follow from it.

\emph{Decomposition.} The index arithmetic by itself costs 23\,ns at $d{=}13$ (first row, where the ranks are computed and nothing is read). A query against the array in anonymous memory costs 81\,ns, so a query consists of the index plus one memory access of roughly 60\,ns. The file-mapped variant pays a further 17\,ns for 4\,KB pages and minor faults, which is the cost of being able to exceed RAM and has nothing to do with the index; the
\texttt{-nohuge} controls in the artifact put the transparent-huge-page
component at 10--12\,ns for both this structure and the flat table.

\emph{The strongest simple baseline.} An open-addressing table with inline
keys is faster to query than the tapered array once both exceed cache:
54\,ns against 81\,ns, a factor of $1.5$. The decomposition says why.
Both structures make one memory access per query, and the flat table's
hash costs about 2\,ns where the rank costs 23. The values-only array therefore pays its $O(d)$ index on top of the same single access, and against an inline-key flat table that shows up as a deficit of 12--27\,ns across the configurations of Tables~\ref{tab:headtohead} and~\ref{tab:desksweep}. The array's advantage over every key-storing structure lies in memory and construction. It uses $4.0\times$ less peak RSS than the flat table here,
$7.2\times$ less counting only the structures (a 143\,MB array against a
1024\,MB table), $5.7\times$ less than \texttt{unordered\_map}, and the
gap widens without bound in $N$ (\S\ref{sec:scaling}). It builds
$100\times$ faster than the flat table, because a values-only array is
filled sequentially by rank while any hash table must enumerate and insert
its keys.

\emph{The chained hash.} \texttt{std::unordered\_map} pays at least two
dependent memory accesses plus a $d$-tuple hash. The array beats it by $2.9\times$ here and by $1.1$--$1.8\times$ on the server
(\S\ref{sec:scaling}). We retain it because it is the container most practitioners would try first; we added the flat table because it is the comparison a careful reader would require.

\emph{The original index.} The table-free rank matches the node-count
index of \cite{biswas2015} on query latency (the small residual comes from the pinned sub-range's better locality rather than from the method) while dispensing
with its per-instance preprocessing and its pinned-coordinate restriction.
The gain over the original lies in construction and generality; query speed is unchanged.

\begin{table}[ht]
\caption{Desktop sweep at fixed $d{=}10$, 2M queries: query latency (ns, median of three runs) and peak RSS (MB), array in RAM. Structure sizes: array $4N$ bytes; flat
table 16 bytes per slot at the stated capacity.}
\label{tab:desksweep}
\centering\small
\begin{tabular}{rcccccc}
\toprule
$N$ & \multicolumn{2}{c}{tapered} & \multicolumn{2}{c}{flat table} & \multicolumn{2}{c}{\texttt{unordered\_map}} \\
 & ns & MB & ns & MB & ns & MB \\
\midrule
92K & 19.7 & 79 & 23.5 & 81 & 39.0 & 83 \\
1.14M & \textbf{25.2} & 83 & 43.5 & 111 & 131.5 & 126 \\
8.44M & 63.0 & 111 & \textbf{51.0} & 335 & 164.8 & 422 \\
44.4M & 76.9 & 248 & \textbf{55.3} & 1103 & 185.9 & 1879 \\
\bottomrule
\end{tabular}
\end{table}
\begin{figure}[ht]
\centering
\includegraphics[width=\linewidth]{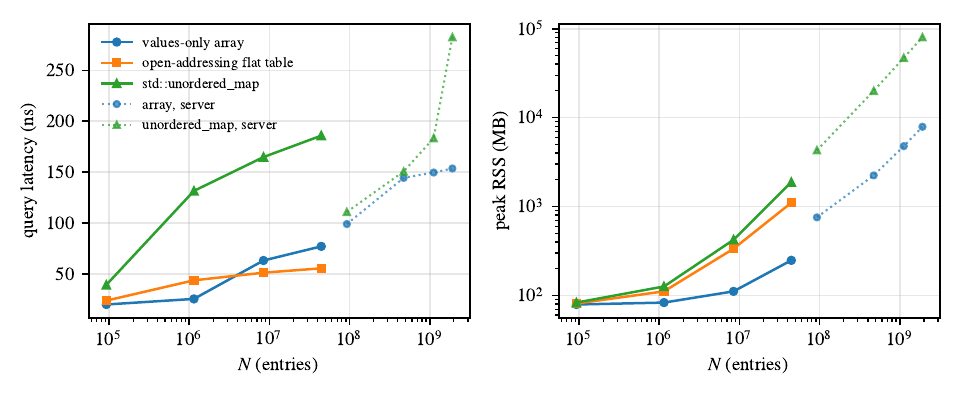}
\caption{Query latency (left) and peak RSS (right) against $N$ at fixed
$d{=}10$. Solid lines: desktop, array in RAM, median of three runs; the
RSS floor of about 80\,MB is the query buffer. Dotted lines: the server
sweep of Table~\ref{tab:scaling}, which measured only the array and the
chained map. The array is fastest while it fits cache and the flat table
does not, then settles 12--27\,ns behind it; its memory stays a widening
factor below both.}
\label{fig:scaling}
\end{figure}

Table~\ref{tab:desksweep} and Figure~\ref{fig:scaling} add a cache-band
observation. At $N{=}1.14$M
the $4.6$\,MB array fits this machine's 8\,MB L3 while the 32\,MB flat
table does not, and the array is then $1.7\times$ faster. Being
$7\times$ smaller is a latency advantage wherever it moves the structure
up one level of the memory hierarchy. We report this as a regime rather than as a general claim. Beyond it the flat table's advantage, 12\,ns at $N{=}8$M and 22--27\,ns at larger $N$, is the steady state.

We do not report a spatial-index baseline. A $k$-d tree or
similar structure answers \emph{nearest-neighbour} queries, which is a
strictly harder problem than the one we face: because the grid is a
product of per-coordinate spacings, the nearest gridpoint is obtained in
closed form by rounding each coordinate independently, in $O(d)$, with no
search. Benchmarking against a structure built for unstructured point sets
would therefore measure the cost of generality we do not need, and at
$d=12$--$15$ such indices degrade towards linear scan in any case. The
relevant comparison is the one we make: against containers that answer the
same exact-key lookup.

\subsection{\texorpdfstring{Scaling in $N$}{Scaling in N}}\label{sec:scaling}
We fix the query dimension at $d{=}10$ and sweep $B$, so that only $N$ changes, with 10M queries per point. The sweep reaches $N{=}1.92$B on the server, $51\times$ the largest in-RAM configuration of Table~\ref{tab:headtohead}. Table~\ref{tab:scaling} reports it:

\begin{table}[ht]\caption{In-RAM scaling, fixed $d{=}10$, server.}
\label{tab:scaling}\centering\small
\begin{tabular}{lccccc}
\toprule
$N$ & tapered ns & hash ns & lat.\ ratio & tapered RSS & hash RSS \\
\midrule
92.6M  & 98.9  & 110.8 & 1.12$\times$ & 0.74\,GB & 4.19\,GB \\
473M   & 144.3 & 150.6 & 1.04$\times$ & 2.19\,GB & 19.5\,GB \\
1.12B  & 149.6 & 183.4 & 1.23$\times$ & 4.66\,GB & 45.8\,GB \\
1.92B  & 153.6 & 282.8 & 1.84$\times$ & 7.70\,GB & 78.4\,GB \\
\bottomrule
\end{tabular}
\end{table}

The memory ratio grows monotonically and without bound: it rises from $5.7\times$ to $10.2\times$ over this range, and is unbounded in
principle, since the array carries no keys. Query latency against the chained map requires more care. At fixed $d$ the tapered cost is a single $O(d)$ computation plus one memory access and stays essentially flat in $N$ ($99\!\to\!154$\,ns), while the chained map grows ($111\!\to\!283$\,ns) as it becomes DRAM-bound, so the ratio widens with $N$ at fixed $d$ (to $1.8\times$ here). The inline-key flat table of Table~\ref{tab:headtohead} was measured only on the desktop; since it also makes one access per query at any $N$, its latency should track the array's less the index cost at every point of this sweep, and its footprint of about 24 bytes per entry at load $0.66$ would sit $4$--$7\times$ above the array's throughout. The absolute ratio is machine-dependent
($1.1$--$1.8\times$ across our two platforms) and noisy on this shared machine (tapered latency showed $20$--$30\%$ run-to-run variance under a concurrent
tenant; reported as the median of 3 runs, hash single runs). The chained map never exhausted memory on this 499\,GB machine ($78$\,GB peak); the regime in which a key-storing structure cannot be built is the out-of-core experiment on the desktop (\S\ref{sec:ooc}) rather than this sweep.

\subsection{\texorpdfstring{Out of core: $5.57$ billion entries on a 16\,GB desktop}{Out of core: 5.57 billion entries on a 16 GB desktop}}
\label{sec:ooc}
The configuration is $B{=}22$, $d{=}15$, giving $N=\binom{36}{15}=5{,}567{,}902{,}560$ entries and a 22.27\,GB value file, $1.4\times$ the machine's RAM, on a rotational disk, the harshest storage case. The sequential build took 235\,s (24\,M entries/s, disk-bound).
A full-file scan verified \emph{every} payload (0 mismatches) at 28\,M entries/s. This is the streaming-aggregate mode of use, roughly $10^3$ times the random-access rate, with every argument implicit in its index. Random point
queries (20{,}000, \texttt{mmap} + \texttt{madvise(RANDOM)}): mean
17.0\,ms, p50 15.5\,ms, p90 25.1\,ms; 89.8\% of queries went to disk with
mean 18.9\,ms, consistent with one seek per query. The $O(d)$ index arithmetic is not measurable, and there is no second dependent access. For contrast, the chained map's measured 43.5\,B/entry footprint would require $\approx$242\,GB of RAM at this $N$ ($15\times$ the machine), and the inline-key flat table at its measured 24\,B/entry would need $\approx$134\,GB ($8\times$); an on-disk chained layout would pay $\ge2$ dependent seeks per query on a $4.5\times$ larger file (hence a $\sim4.5\times$ smaller cached fraction), and an on-disk flat table would pay one seek on a $6\times$ larger file. In this regime the comparison is no longer about speed: the values-only array is the only one of the three that can be built.

On the server's enterprise NVMe (same query protocol, $N{=}1.12$B,
$4.48$\,GB file, pages evicted via \texttt{posix\_fadvise(DONTNEED)}), a cold random query costs p50 67\,\textmu s, mean 71\,\textmu s;
an immediate warm re-run collapses to $\approx$1\,\textmu s
($\approx$70$\times$), confirming that the cold pages were indeed on disk. Against the HDD's 15.5\,ms p50 this is a roughly $230\times$ media speedup for the identical
one-access-per-query structure. Combining with
\S\ref{sec:parallel}--\S\ref{sec:scaling} yields a three-tier first-touch latency hierarchy at $N\!\sim\!10^9$: resident array
$\approx$100--280\,ns; page-cache-resident file, fresh \texttt{mmap}
(minor fault only) $\approx$1\,\textmu s; evicted, true NVMe read
$\approx$67--85\,\textmu s. Each tier is one access, and the taper's $O(d)$ arithmetic is never the bottleneck.

\paragraph{Locality under a perturbation workload.}
Lemma~\ref{lem:prefix} predicts that queries which differ in a tail
coordinate land on the same page and queries which differ in a leading
coordinate do not, while a hash-addressed layout scatters both. We tested
this on the same 22.27\,GB file with a random walk in key space: from a
random sorted key, each step changes one coordinate, chosen uniformly, by
$\pm1$ subject to monotonicity, for 20{,}000 steps. The same walk was
addressed two ways, by the lexicographic rank and by a 64-bit mix of the key, as any hash-addressed layout would. The file's pages were
evicted before each run (0.00\% resident, checked with \texttt{mincore}),
and a hit was decided by calling \texttt{mincore} on the target page before the access rather than by timing it. Table~\ref{tab:walk} gives the result.

\begin{table}[ht]
\caption{Random-walk queries on the 22.27\,GB out-of-core file, 20{,}000
steps, page cache cold at start. Hit rates by the coordinate the step
changed (1 = leading); ``same page'' is the fraction of steps whose target
page equals the previous step's. Two seeds for the lexicographic walk
agreed to within a point; seed 42 shown. Overall: hit rate 40.4\% against 6.7\%, mean latency 3.9\,ms against 8.2\,ms, median 0.24\,ms against 8.1\,ms.}
\label{tab:walk}
\centering\footnotesize
\begin{tabular}{lrrrrrrrr}
\toprule
& \multicolumn{8}{c}{coordinate changed} \\
& 1--8 & 9 & 10 & 11 & 12 & 13 & 14 & 15 \\
\midrule
lex.\ hit \% & 11 & 18 & 34 & 65 & 93 & 99 & 100 & 100 \\
lex.\ same page \% & 0 & 8 & 25 & 61 & 92 & 98 & 100 & 100 \\
hashed hit \% & 7 & 7 & 8 & 7 & 6 & 7 & 6 & 7 \\
hashed same page \% & 0 & 0 & 0 & 0 & 0 & 0 & 0 & 0 \\
\bottomrule
\end{tabular}
\end{table}

\begin{figure}[ht]
\centering
\includegraphics[width=.8\linewidth]{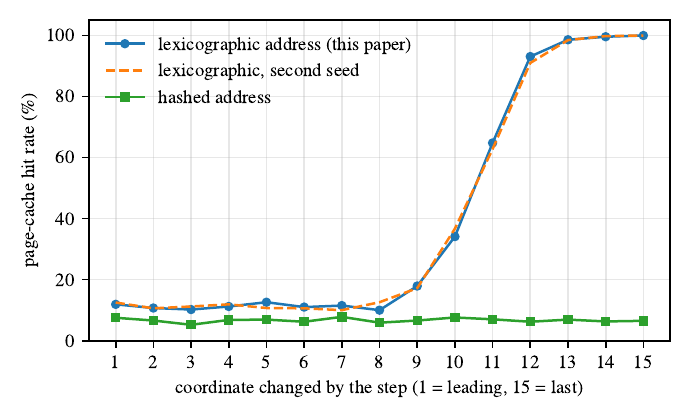}
\caption{Page-cache hit rate of each random-walk step on the 22.27\,GB
out-of-core file, by the coordinate the step changed, for the same walk
addressed lexicographically and by a hash of the key. Cold cache at the
start of each run; hits decided by \texttt{mincore} before the access.}
\label{fig:locality}
\end{figure}

The per-coordinate rows, drawn in Figure~\ref{fig:locality}, show the
lemma directly. Under lexicographic
addressing a step in any of the last four coordinates stays on the same
4\,KB page at least 92\% of the time, a step in coordinate 11 about
three-fifths of the time, and a step in the leading eight coordinates
never; the residual 11\% hit rate there is the cache warmed by earlier
steps. Under hashed addressing every coordinate behaves alike, at 6--8\%, the fraction of the file the page cache happened to hold, and no two
consecutive steps ever share a page. Overall the walk hit the cache
$6\times$ more often (40\% against 7\%), touched $1.6\times$ fewer
distinct pages, and finished in half the wall-clock time. Misses were also cheaper, 4.0--6.6\,ms against 8.8\,ms, because nearby pages on a rotational disk mean shorter seeks. None of this is a property of the workload, which perturbs every coordinate equally often. It is a property of the addressing, which writes the tail coordinates, the ones the taper resolves most coarsely and that matter least to $f$, into the low-order digits of the address.

\subsection{Parallel construction via unranking}\label{sec:parallel}
On a 128-thread dual-EPYC server we fill tables of
$N=\binom{41}{10}\!\approx\!1.12$B and $\binom{46}{10}\!\approx\!4.08$B
entries; worker $j$ owns the index range $[jN/T,(j{+}1)N/T)$ with zero coordination, no locks and no shared state, because $\unrank$ recovers each entry's argument from its index alone. Table~\ref{tab:parallel} reports two variants: a
trivial-payload \texttt{flat} fill, and \texttt{unrank}, which reconstructs the column vector for each entry (the realistic case). Every build
was verified on $10^6$ random indices (0 errors across all 32 runs).

\begin{table}[ht]\caption{Zero-coordination parallel build,
$N{=}1.12$B ($B{=}32,d{=}10$), \texttt{/dev/shm}. \texttt{unrank} is the
argument-reconstructing variant.}\label{tab:parallel}
\centering\small
\begin{tabular}{lccc}
\toprule
$T$ & \texttt{flat} entries/s & \texttt{unrank} entries/s & \texttt{unrank} speedup \\
\midrule
1  & 1.20e9 & 2.31e8 & 1.00$\times$ \\
4  & 1.02e9 & 6.25e8 & 2.70$\times$ \\
16 & 9.17e8 & 8.86e8 & 3.83$\times$ \\
64 & 9.55e8 & 9.59e8 & 4.15$\times$ \\
\bottomrule
\end{tabular}
\end{table}

The \texttt{flat} fill is already memory-bandwidth-bound at a single thread
($\approx$5\,GB/s of sequential stores), so additional threads help little; the realistic \texttt{unrank} variant is compute-bound at low $T$ and scales to
$\approx$4.1$\times$ by $T{=}64$, at which point it converges onto the same $\approx$9.5e8 entries/s ceiling as \texttt{flat}. With enough threads the $O(d)$ reconstruction is hidden entirely behind the bandwidth limit. The ceiling on this two-socket machine is DRAM bandwidth rather than thread count.
Peak RSS was $3.8$ bytes/entry (the 4-byte payload; \emph{no} per-entry key
overhead).

\emph{Comparison with the hash.} A single \texttt{std::unordered\_map} builds the same $N{=}1.12$B content in $234$\,s using $43.1$\,GB; the parallel hash (16 shard maps, 16 threads, one mutex per shard) achieves only $200$\,s, a $1.17\times$ speedup, with lock and node-allocation contention dominating. The
flat array builds the same content in $0.9$\,s (trivial payload, $T{=}1$) to
$4.9$\,s (argument-reconstructing \texttt{unrank}, $T{=}1$), and in
$1.2$\,s at $T{=}64$: a $48$--$250\times$ build-time advantage. Its peak RSS is $4.2$\,GB against the hash's $43.1$\,GB, a $10.2\times$ memory advantage. Both are build-only peaks, so they sit slightly below the build-plus-query peaks of Table~\ref{tab:scaling}, which are $4.66$\,GB and $45.8$\,GB at the same $N$ and give $9.8\times$. The widest ratio we measure, $10.2\times$, is the $N{=}1.92$B row of that table.
The memory gap is wider than the $5.7\times$ seen at
desktop scale because the hash's fixed per-entry overhead grows with $N$.

\subsection{End-to-end applications}
Reported in Section~\ref{sec:apps}: Plackett--Luce memoization is
$25$--$55\times$ faster than exact Newton at $2.6\cdot10^{-3}$ mean error
(Table~\ref{tab:pl}), with measured bias $0.30$--$0.47\times$ the mean
absolute error on the three finer grids, confirming Lemma~\ref{lem:bias}; $\alpha$-entmax thresholds are
a negative result ($92$\,ns memo vs.\ $35$\,ns exact).

\section{Applications}\label{sec:apps}

\subsection{Plackett--Luce normalization: where memoization wins}
\label{sec:app-pl}
The canonical $f$ of \cite{biswas2015} takes a sorted $x$ with $x_1{=}0$ and returns the $p_1$ that solves $\sum_i p_1^{1/(1-x_i)}=1$, the top-choice probability of a Luce/Plackett model \cite{luce1959,plackett1975}. No closed form exists, so every exact evaluation is an iterative root-find, each iteration of which costs $\ell$ transcendental \texttt{pow} calls. We draw $2\cdot10^5$ queries from the generator of \cite[\S5]{biswas2015} ($b{=}5$, $\ell{=}32$) and compare single-query evaluation in C++ at \texttt{-O3}.

\begin{table}[ht]\caption{Plackett--Luce memoization
($\gamma{=}2$ taper). Speedups are over Newton / bisection.}
\label{tab:pl}\centering\small
\begin{tabular}{lccccc}
\toprule
grid & table & mean err & bias/mean & memo & speedup \\
\midrule
$B{=}12,d{=}8$   & 0.1\,MB & 1.30e-2 & 0.99 & 70\,ns  & $\mathbf{55\times}$ / 406$\times$ \\
$B{=}16,d{=}12$  & 31\,MB  & 3.04e-3 & 0.34 & 117\,ns & $\mathbf{33\times}$ / 244$\times$ \\
$B{=}20,d{=}12$  & 218\,MB & 2.64e-3 & 0.47 & 145\,ns & $\mathbf{26\times}$ / 199$\times$ \\
$B{=}16,d{=}16$  & 620\,MB & 2.65e-3 & 0.30 & 149\,ns & $\mathbf{25\times}$ / 189$\times$ \\
\bottomrule
\end{tabular}
\end{table}

Exact evaluation costs $3.8$\,\textmu s (Newton) or $28$\,\textmu s
(bisection) per query; the memo table answers in $70$--$149$\,ns, a
$25$--$55\times$ speedup over Newton and up to $406\times$ over bisection, against the $\approx$10$\times$ reported in \cite{biswas2015}.
Accuracy saturates near $2.6\cdot10^{-3}$ mean absolute error (the
interpolation floor for nearest-grid rounding; \cite{biswas2015}'s
credit-weighted variant reaches $\approx$$10^{-3}$). The measured bias is $0.30$--$0.47\times$ the mean absolute error on the three finer grids, and an order smaller still for the smoother $f$ of \S\ref{sec:app-entmax}. On the coarsest grid ($B{=}12$, $d{=}8$) it is $0.99\times$, because the $\delta^{2}$ regime has not yet been entered. This is the empirical counterpart of Lemma~\ref{lem:bias}, and it is what justifies the aggregate use case.

\subsection{\texorpdfstring{$\alpha$-entmax thresholds: where it does \emph{not}}{alpha-entmax thresholds: where it does not}}
\label{sec:app-entmax}
It is equally important to report where the structure does not pay.
$\alpha$-entmax \cite{peters2019} computes a threshold $\tau$ over sorted
logits, and accelerating it is the subject of AdaSplash
\cite{goncalves2025}; it is evaluated per row$\times$head$\times$layer
$\times$token, and its domain appears ideal, since entries below $\tau$ have zero influence, so the tapered tail is exact. Our test set is $199{,}741$ pre-softmax attention rows sampled
uniformly over (layer, head, position) from GPT-2 small run on
WikiText-2 (raw), keeping each row's top-$32$ sorted scores; on these the
memo table is accurate ($\tau$ error $3.8\cdot10^{-3}$ mean, bias only $0.005\times$ that) but is slower: $92$\,ns against $35$\,ns for the exact solver.

The reason is the size of the support. For $\alpha{=}1.5$ the threshold has a closed
form over the \emph{active support}, and on real attention rows that support is tiny: median $2$, mean $4.0$, with only $0.88\%$ of rows
reaching the top-32 window. An $O(\text{support})$ closed form cannot be outperformed by an $O(d)$ interpolation plus an $O(d)$ rank plus a memory access. (The memo table does outperform a naive $60$-iteration bisection by $6.4\times$, but that is not a baseline anyone should use.) We note this because a widely quoted
$37.6$\,\textmu s/row figure for single-row \texttt{entmax15} reflects Python and tensor per-call overhead; the same mathematics in C++ takes $35$\,ns.

\paragraph{Scoping rule.} Memoize when the exact evaluation is iterative and transcendental, as for Plackett--Luce at $3.8$\,\textmu s, and not when it is merely ``a solve'', as for entmax at $35$\,ns. Sortedness, rank-decaying influence and an exact tail are necessary but not sufficient; the exact
baseline must itself be expensive.

\section{Related Work}\label{sec:related}

\paragraph{Ranking combinatorial objects.} Assigning consecutive integers
to combinatorial objects in lexicographic order originates with Lehmer's
combinatorial number system \cite{lehmer1964} and was made algorithmic by
Buckles and Lybanon \cite{buckles1977}; Knuth \cite[\S7.2.1.3]{knuth2011}
and Kreher and Stinson \cite{kreher1999} give the modern treatment,
including combinations with repetition via the stars-and-bars bijection,
and Genitrini and P\'epin \cite{genitrini2021} revisit unranking. In its
general form the principle is Cover's enumerative source coding
\cite{cover1973}: an index is the count of admissible continuations. Our
contribution is not this machinery but its identification with the tapered memoization grid, which removes that scheme's tables, extends it to variable per-depth alphabets and supplies the inverse it lacked.

\paragraph{Succinct structures and perfect hashing.} Storing data at the
information-theoretic minimum with fast access is the subject of succinct
and implicit data structures \cite{jacobson1989,raman2002}; rank-indexed
blocks are the device inside \textsc{rrr}-style encodings. For arbitrary
static key sets, retrieval structures achieve $(1+o(1))Nw$ bits without
storing keys \cite{pagh2001,dietzfelbinger2008}, at the price of
randomized construction and no order or inverse. Minimal perfect hashing
carries $\Omega(N)$-bit description lower bounds
\cite{fredman1984,mehlhorn1982}, with tight bounds known for monotone
variants \cite{belazzougui2009,assadi2023} and $\Theta(N\log N)$ for
arbitrary prescribed orders \cite{fox1991}. Our key set escapes these bounds because it is a single parameterized family rather than an arbitrary subset.

\paragraph{Grids for functions with decaying coordinate importance.}
Sparse grids \cite{bungartz2004}, and their dimension-adaptive \cite{gerstner2003} and anisotropic \cite{nobile2008,nobile2016} variants, allocate resolution per dimension according to importance. Griebel and Knapek \cite{griebel2000} determine optimal approximation spaces from a ratio of smoothness exponents, together with a penalty for
choosing that ratio wrongly, which is the closest antecedent of our
Theorem~\ref{thm:penalty} (theirs is one-sided and additive in the order,
on a sparse grid; ours is two-sided, multiplicative on $\log$-cost, at
finite $\eps$, for a full tapered tensor grid). That grid size can be made
independent of the ambient dimension under summable coordinate importance
is the classical weighted-space and tractability theory
\cite{sloan1998,novak2008,gnewuch2011}, with matching $N$-term results for
anisotropic analytic classes \cite{cohen2011} and the $\log^2(1/\eps)$
entropy regime going back to Kolmogorov and Tikhomirov
\cite{kolmogorov1959}. We use this regime and do not extend it. Symmetry alone does not give dimension-independence \cite{weimar2014}; the decay does.

\paragraph{Balancing competing geometric rates.} The phenomenon behind
Theorem~\ref{thm:penalty} is classical: a geometric schedule must match a
geometric target rate, with graceful but real degradation otherwise.
It is the optimal bit-allocation (reverse water-filling) result
of transform coding \cite{huang1963,berger1971,gersho1992}, and it appears
sharply in rational approximation, where the achievable exponent is the
minimum of two competing rates and is maximized when they coincide
\cite{gopal2019,xiang2023}; exponential-sum \cite{braess2005} and
$hp$-refinement analyses have the same structure. Our contribution is a closed form at finite $\eps$ for this construction; the effect itself is not new.

\paragraph{Approximating softmax-family transforms.} The
$\alpha$-entmax transform \cite{peters2019} and its sparse-attention
descendants \cite{goncalves2025} motivate our second application; adaptive truncation
schemes for decoding \cite{nguyen2025} share the sorted-logit domain.
Lookup-table approximation of softmax is well established in inference
hardware, but for scalar transforms, whereas the setting here has a vector argument, monotone and tapered. Our finding for this family is negative and, we believe, useful: when the exact solver is itself
$O(\text{support})$, tabulation does not pay.

\section{Conclusion}\label{sec:conclusion}
A memo table for a function of a sorted argument need not store its keys,
and the version that does not is smaller, faster to build, and usable at sizes where the alternatives cannot be constructed.
Recognizing the tapered memoization grid of \cite{biswas2015} as an
instance of the combinatorial number system turns its index into a
closed-form $O(d)$ computation, removes its preprocessing tables, extends
it to variable per-depth alphabets, and, most importantly, supplies the inverse it lacked. With unranking available the keys need never be
stored, and the memo table becomes a flat array of values addressed by
content. We set out to measure the practical consequences; they are as follows. The footprint is a factor $5.7$--$10.2\times$ below a chained hash memo and $4$--$7\times$ below an inline-key flat table, and both factors widen with $N$. On query latency the array outperforms the chained map and trails the flat table by the cost of its own index, about 23\,ns at $d{=}13$; that is the price of storing no keys, and we have measured it rather than estimated it. The table is built by 64 threads without any locking, at the machine's memory-bandwidth limit, where a sharded hash manages $1.17\times$. And at $22$\,GB on a $16$\,GB desktop it answers each query in exactly one disk access, at a size where no key-storing container can be constructed.

The applications set the limits of the claim as well as supporting it. Memoizing the Plackett--Luce normalization, an iterative transcendental solve, is $25$--$55\times$ faster than Newton's method, beyond the
$\approx$$10\times$ originally reported. Memoizing the $\alpha$-entmax threshold is slower than solving it exactly, because its active support is tiny. The lesson is that a sorted argument with
rank-decaying influence and an exact tail is not sufficient reason to
tabulate: the exact evaluation must itself be expensive.

Several directions remain. A GPU-resident variant would need the index
arithmetic mapped onto warp-level primitives and a batched interpolation
kernel. The taper schedule is geometric by assumption here; learning a
non-geometric schedule from an empirical influence profile, guided by
Theorem~\ref{thm:penalty}, may recover part of the accuracy floor we
observe. Hybrids with retrieval structures could extend key-free storage
to grids whose admissible set is not a clean lattice. Finally, that
accuracy floor comes from nearest-grid rounding; the credit-weighted
interpolation of \cite{biswas2015} reaches roughly $10^{-3}$, and
combining it with table-free indexing is straightforward future work.


\section*{Acknowledgements}
The author thanks Kenneth W. Regan, his co-author on the 2015 paper on
which this work builds.

\appendix

\section{\texorpdfstring{Derivation of Theorem~\ref{thm:penalty}}{Derivation of the penalty theorem}}\label{app:penalty}

\subsection{The water-filling optimum and its depth bonus}
Minimize $\sum_i \ln n_i$ subject to $\sum_i g_i/(2n_i)\le\eps$,
$n_i\ge1$. Substituting $t_i=\ln n_i$ makes the constraint convex in $t$
and the objective linear, so the Karush--Kuhn--Tucker conditions are necessary and sufficient:
on the active set, per-level errors equalize, $g_i/(2n_i)=e^\ast$; this
is the classical reverse water-filling / bit-allocation optimum
\cite{huang1963,berger1971,gersho1992}. With $g_i=r^{\,i-1}$, level $i$
is active iff $g_i>2e^\ast$, giving an active prefix of length $k$
satisfying $r^{k-1}>2e^\ast\ge r^{k}$. Inactive levels contribute the
tail $\sum_{i>k}g_i/2=r^{k}/(2(1-r))$, so
$e^\ast=\big(\eps-\tfrac{r^{k}}{2(1-r)}\big)/k=\tfrac{\eps}{k}(1+O(1/k))$,
and the activation boundary gives
\begin{equation}\label{eq:kfix}
k\,a \;=\; \ln\frac{1}{2e^\ast} + O(1)\;=\;L' + \ln k + O(1),
\end{equation}
the \emph{depth bonus}: because $e^\ast=\eps/k$ decreases as levels
activate, the optimum runs $\ln k$ deeper (in budget units) than a
naive $\eps$-per-level account suggests. Summing,
\[
\ln|U^\ast|=\sum_{i\le k}\Big[\ln\tfrac{1}{2e^\ast}-(i-1)a\Big]
= k\ln\tfrac{1}{2e^\ast}-a\tfrac{k(k-1)}{2}
= \tfrac{a}{2}k^{2}+O(k).
\]

\subsection{\texorpdfstring{Geometric schedule, $s<r$ (refining too fast)}{Geometric schedule, s < r (refining too fast)}}
Per-level errors $\frac{1}{2n_1}(r/s)^{i-1}$ \emph{increase} toward the
truncation depth $m$ (where $n_m=\bar c\in[1,1/s)$), so the total is
dominated by the last active level plus the tail:
$\big(\tfrac{r^{m-1}}{2}\big)\,C_< \,(1+o(1))$ with
$C_<=\frac{1}{\bar c(1-s/r)}+\frac{r}{1-r}$. Feasibility pins
$r^{m-1}=2\eps/C_<$, i.e.
\begin{equation}\label{eq:mfix-lt}
m\,a = L' + \ln C_< + O(1).
\end{equation}
With $\ln n_1=(m-1)b+\ln\bar c$,
$\ln|U_s|=\sum_{i\le m}\big[\ln n_1-(i-1)b\big]
=\tfrac{b}{2}m^{2}+O(m)$.

\subsection{\texorpdfstring{Geometric schedule, $s>r$ (refining too slowly)}{Geometric schedule, s > r (refining too slowly)}}
The error series converges from the first term:
$\sum_i \frac{1}{2n_1}(r/s)^{i-1}=\frac{1}{2n_1}\cdot\frac{1}{1-r/s}
\,(1+o(1))$, so $n_1=C_>/(2\eps)$ with $C_>=\frac{1}{1-r/s}$, i.e.\
$\ln n_1=L'+\ln C_>$, and the truncation depth is $m=1+\ln n_1/b$. The
tail is negligible since $am=(a/b)(L'+\ln C_>)>L'$. Again
$\ln|U_s|=\tfrac{b}{2}m^{2}+O(m)$.

\subsection{Combining}
\[
\frac{\ln|U_s|}{\ln|U^\ast|}
=\frac{\tfrac{b}{2}m^{2}+O(m)}{\tfrac{a}{2}k^{2}+O(k)}
=\frac{b}{a}\cdot\Big(\frac{m}{k}\Big)^{2}\big(1+O(1/L')\big),
\]
and substituting \eqref{eq:kfix} and \eqref{eq:mfix-lt} (resp.\ the
$s>r$ analogue, where $\frac{b}{a}(\frac{m}{k})^2$ produces the
$\frac{a}{b}$ branch) yields \eqref{eq:penalty}. As $\eps\to0$ the
bracket tends to $1$ like $\big(1-\frac{\ln k-\ln C}{L'}\big)^{2}$, logarithmically slowly, which is why the asymptote $\rho_\infty$
overstates the penalty at any practical~$\eps$. \hfill$\square$

\subsection{Numerical verification}
We evaluated \eqref{eq:penalty} against the exact discrete optimum
(computed by numerical water-filling) over $33$ cells:
$r\in\{0.25,0.5,0.75\}$, $s\in\{0.25,\dots,0.9\}$,
$\eps\in\{10^{-6},10^{-8},10^{-10}\}$. Representative values at
$\eps=10^{-8}$ are given in Table~\ref{tab:penaltyverify}.

\begin{table}[ht]\caption{Predicted versus exact penalty. $\rho_\infty$ is
the classical asymptote alone.}\label{tab:penaltyverify}
\centering\small
\begin{tabular}{lcccc}
\toprule
$(r,s)$ & exact & \eqref{eq:penalty} & error & $\rho_\infty$ error \\
\midrule
$(0.75,0.50)$ \emph{halving} & 1.878 & 1.858 & $-1.0\%$ & $+28\%$ \\
$(0.75,0.90)$ & 2.110 & 2.136 & $+1.2\%$ & $+29\%$ \\
$(0.50,0.25)$ & 1.578 & 1.554 & $-1.5\%$ & $+27\%$ \\
$(0.50,0.75)$ & 1.866 & 1.910 & $+2.3\%$ & $+29\%$ \\
$(0.25,0.50)$ \emph{halving} & 1.569 & 1.628 & $+3.8\%$ & $+27\%$ \\
$(0.25,0.75)$ & 3.585 & 3.801 & $+6.0\%$ & $+34\%$ \\
\bottomrule
\end{tabular}
\end{table}

Across all $33$ cells the closed form is within $0.6$--$8.2\%$ (median
$\approx1.5\%$) and improves monotonically as $\eps$ decreases in every
cell (Table~\ref{tab:penaltyverify}), confirming the two-term form; the asymptote alone is off by $+16\%$
to $+43\%$. The residual is the dropped $O(m)$ terms and is largest where
the active depth $k$ is small, i.e.\ for fast decay.

\section{\texorpdfstring{Proof of Lemma~\ref{lem:bias}}{Proof of the bias lemma}}\label{app:bias}
Taylor-expand within the cell of gridpoint $u$:
$f(x)-f(u)=\nabla f(u)\!\cdot\!(x-u)+\tfrac12 (x-u)^{\!\top}H(\xi)(x-u)$
with $\xi$ on the segment. Under an exactly per-cell-uniform query
measure, $\mathbb E[x_i-u_i]=0$ annihilates the first-order term, and
independence across coordinates within the cell gives
$\mathbb E[(x_i-u_i)(x_j-u_j)]=\delta_{ij}\,\delta_i^{2}/12$, so
\[
\mathbb E[f(x)-f(u)]
=\sum_i \frac{\partial^2 f}{\partial x_i^2}(u)\,\frac{\delta_i^{2}}{24}
+o\!\Big(\sum_i\delta_i^{2}\Big).
\]
For a general $C^{1}$ query density $\rho$, the within-cell mean offset
is $\mathbb E[x_i-u_i]=O\big(\delta_i^{2}\,\partial_i\ln\rho\big)$, so
the first-order term also contributes at second order, changing only
the constant. One-sided cells (the monotone boundary and truncated tail levels, where $n_i=1$ and rounding goes to the padded ceiling) retain first-order bias; under the optimal taper those levels have
$g_i\le 2\eps/k$, so their contribution stays within the error budget
but not at the $\delta^{2}$ rate. This is high-resolution quantization
analysis in the Bennett lineage \cite{bennett1948} applied to the
tapered grid; no novelty is claimed for the technique.
\hfill$\square$

\section*{Declarations}

\noindent\textbf{Funding.} This research received no specific grant from any
funding agency in the public, commercial, or not-for-profit sectors.

\medskip\noindent\textbf{Competing interests.} The author declares no
competing financial or non-financial interests. The author is not, and has
not been, an editor or editorial board member of this journal.

\medskip\noindent\textbf{Data availability.} All code, build scripts, raw
measurement logs, and the script that regenerates the $\alpha$-entmax
attention-row dataset are publicly available at
\url{https://github.com/tamalrkm/tapered-grid} under the MIT licence, and
permanently archived on Zenodo with the concept DOI
\texttt{10.5281/zenodo.21675851}, which resolves to the latest release
\cite{artifact2026}. The scripts that generate every figure, table and
reported measurement are included in that archive.

\medskip\noindent\textbf{Ethics approval.} Not applicable. This research
involved no human participants, no animal subjects, and no personal data.

\medskip\noindent\textbf{Patient consent.} Not applicable.

\medskip\noindent\textbf{Clinical trial registration.} Not applicable.

\medskip\noindent\textbf{Permission to reproduce material from other
sources.} Not applicable. All figures and tables are the author's own and
have not been published elsewhere. The attention-row dataset of
\S\ref{sec:app-entmax} is derived from the publicly available WikiText-2
corpus and the GPT-2 model, and is regenerated from scratch by a script in
the artifact rather than redistributed.

\medskip\noindent\textbf{Use of generative AI and AI-assisted technologies.}
During the preparation of this work the author used Claude (Anthropic),
accessed through the Claude Code interface, to draft and revise manuscript
text, to implement the benchmark and analysis programs released with the
artifact, and to execute the experiments reported in
Sections~\ref{sec:experiments} and~\ref{sec:apps}. After using this tool the
author reviewed and edited the content as needed and takes full
responsibility for the content of the publication. All correctness claims
reported here are additionally machine-verified against brute-force
enumeration, and all measurement code and raw logs are publicly available
(Section~\ref{sec:experiments}).

\bibliographystyle{elsarticle-num}
\bibliography{refs}

\end{document}